 \documentclass[smallabstract,smallcaptions]{dccpaper}

\usepackage{epsfig}
\usepackage{citesort}
\usepackage{amsmath}
\usepackage{amssymb}
\usepackage{color}
\usepackage{url}
\usepackage{algorithm}
\usepackage{algorithmicx}
\usepackage{algpseudocode}
\usepackage{caption}
\usepackage{subcaption}

\newlength{\figurewidth}
\newlength{\smallfigurewidth}

\newtheorem{defn}{Definition}
\newtheorem{theorem}{Theorem}
\newtheorem{proof}{Proof}

\begin{document}

\title
{\large
\textbf{Higher Order Context Transformations}
}

\author{%
Michal Vasinek$^{\ast}$, Jan Platos$^{\ast}$\\[0.5em]
{\small\begin{minipage}{\linewidth}\begin{center}
\begin{tabular}{c}
$^{\ast}$VSB-TU Ostrava \\
17. listopadu \\
Ostrava, 708 00, Czech Republic \\
\url{michal.vasinek@vsb.cz,jan@platos.vsb.cz} \\[0.5em]
\end{tabular}
\end{center}\end{minipage}}
}

\maketitle
\thispagestyle{empty}

\begin{abstract}
The context transformation and generalized context transformation methods, we introduced recently, were able to reduce zero order entropy by exchanging digrams, and as a consequence, they were removing mutual information between consecutive symbols of the input message. 
These transformations were intended to be used as a preprocessor for zero-order entropy coding algorithms like Arithmetic or Huffman coding, since we know, that especially Arithmetic coding can achieve a compression rate almost
of the size of Shannon's entropy.

This paper introduces a novel algorithm based on the concept of generalized context transformation, that allows transformation of words longer than simple digrams. The higher order contexts are exploited using recursive form of a generalized context transformation.
It is shown that the zero order entropy of transformed data drops significantly, but on the other hand, the overhead given by a description of individual transformations increases and it has become a limiting factor in a successful transformation of smaller files.
\end{abstract}

\Section{Introduction}

Arithmetic coding\cite{wit87} is a very successful compresison method that is able to achieve compression rate almost of the size of the limit given by Shannon's formula. In our work we rely on this fact and we assume that Shannon's entropy of transformed data is approximately achievable.
In our previous work about context transformations \cite{vasinek14} and generalized context transformations \cite{vasinek15} we studied, under what conditions the exchange of two different digrams, beginning with the same symbol, leads to the reduction of Shannon's entropy.

In the present paper we shall introduce a modified version of the generalized context transformation algorithm that exploits the fact, that during digrams exchange process, we get an information about positions of particular digram in input message. These positions are later used to form conditional distributions describing occurences of symbols following the transformation digrams and the process of the search for entropy reducing transformations is extended to trigrams and in the same way also for other longer words.

%The rest of the paper is organized as follows. Section \ref{sec:contextTransform} contains description of the generalized context transformation and presents its basic properties. Section \ref{sec:H-red} analyses under what circumstances context transformation reduces entropy. Section \ref{sec:algorithm} describes context transformation algorithm and explains several used concepts including inverse transformation. Section \ref{sec:results} contains summary of achieved results. Last Section \ref{sec:conclusion} concludes the paper and discusses achieved results. 

\Section{Transformations}

The generalized context transformation is an exchange of two different digrams $\alpha\beta$ and $\alpha\gamma$ throughout the input message $m$ over the alphabet $\Sigma$. The exact definition is as follows:

\begin{defn}
\label{def:gct}
Generalized context transformation(GCT) is a mapping $GCT(\alpha\beta\leftrightarrow\alpha\gamma,m):\Sigma^{n}\rightarrow\Sigma^{n}$,
$\Sigma$ is the alphabet of the input message $m$ and $n$ is the length of the input message, that exchanges all digrams $\alpha\beta$ for digram $\alpha\gamma$ and vice-versa.
\end{defn}

We were considering also transformations of the more general form $\alpha\beta\leftrightarrow\gamma\delta$, but in such case, the space, where the transformation would be searched, is of order $|\Sigma|^4$, meanwhile our preferred space, the space of generalized context transformations, is only of order $|\Sigma|^2$. It is computationally much more efficient to work with GCT. 

GCT can be applied on the message from left to right $GCT_{\rightarrow}$ or from right to left $GCT_{\leftarrow}$, even though both transformations seem to be the same and in most cases they provide the same entropy reduction, they differ when we deal with the transformation of the form $GCT(\alpha\alpha\leftrightarrow\alpha\beta)$. Meanwhile the value of zero order entropy change $\Delta H$ can be precisely predicted when $GCT_{\leftarrow}$ is applied, the opposite direction transformation $GCT_{\rightarrow}$ depends also on occurences of $n-grams$ consisting of $n$ repetition of a symbol $\alpha$. 

For instance suppose that the initial message is $m=\alpha\alpha\alpha$ and we apply the transformation $GCT_{\leftarrow}$ first, leaving $m' = GCT_{\leftarrow}(\alpha\alpha\leftrightarrow\alpha\beta, \alpha\alpha\alpha) = \alpha\beta\beta$, in the initial state there were exactly two digrams $\alpha\alpha$ and three symbols $\alpha$. In the final state both digrams $\alpha\alpha$ were replaced by $\alpha\beta$ and we have a number of symbols $\beta$ equal to the number of digrams $\alpha\alpha$ in the initial state of the message. The opposite direction transformation leaves $m' = GCT_{\rightarrow}(\alpha\alpha\leftrightarrow\alpha\beta, \alpha\alpha\alpha) = \alpha\beta\alpha$ and without a knowledge of a distribution of n-grams of the type $\alpha^n$ we are unable to precisely predict the number of newly introduced occurences of $\beta$.

Our algorithm uses a right to left form of transformation $GCT_{\leftarrow}$, it allows us to precisely predict how the zero order entropy will change before the arbitrary generalized context transformation is applied. Inverse transformation to $GCT_{\leftarrow}$ is $GCT_{\rightarrow}$ and vice versa. When a transformation is applied in one direction then the same transformation applied in the opposite direction is its inverse\cite{vasinek15}.

\Section{Shannon's entropy reduction}
\label{sec:H-red}

The zero order or Shannon's\cite{shannon1948} entropy of random variable $X$ is defined  by:

\begin{equation}
H(X)=-\sum_{x\in\Sigma}p(x)\log p(x)
\end{equation}

where, $p(x)$ is a probability of a symbol $x$ in input message and $log(x)$ represents logarithm of $x$ to the base 2. When we use a term entropy, we always mean Shannon's entropy. Suppose that the transformation $GCT_{\leftarrow}(\alpha\beta\leftrightarrow\alpha\gamma)$ is applied. Only probabilities of symbols $\beta$ and $\gamma$ will change. The exact change of probabilities $p_i(x)$, where $i=0$ means the initial state probabilities and $i=1$ the state after the GCT is applied, are expressed by:

\begin{equation}
p_{1}(\beta) = p_{0}(\beta) + p_{0}(\alpha,\gamma) - p_{0}(\alpha,\beta)
\end{equation}

\begin{equation}
p_{1}(\gamma) = p_{0}(\gamma) + p_{0}(\alpha,\beta) - p_{0}(\alpha,\gamma)
\end{equation}

The resulted change of entropy $\Delta H$ is then expressed as a difference of entropies between a final and initial state:

\[
\Delta H = H_{1} - H_{0} 
\]

\begin{equation}
\label{eq:dH}
\Delta H = -\sum_{x\in\{\beta,\gamma\}}p_{1}(x)\log p_{1}(x) + \sum_{x\in\{\beta,\gamma\}}p_{0}(x)\log p_{0}(x)
\end{equation}

Any transformation with $\Delta H$ negative reduces zero order entropy. This zero order entropy reduction has an impact on the values of other information theoretic\cite{Cover} quantities, especially the mutual information drops together with zero order entropy, since the mutual information is given as a relative entropy between the distribution of symbols $p(x)$ and the conditional distribution $p(x|\alpha)$ and it can be interpreted as their distance. The purpose of transformations is to decrease the average distance, given by relative entropy, between the distribution of symbols and all other conditional distributions. Generalized context transformation that reduces entropy is described by the following theorem:

\begin{theorem}
\label{thm:reduction}
Suppose the generalized context transformation $GCT(\alpha\beta\leftrightarrow\alpha\gamma)$. Let $p_{0}(\beta)$ and $p_{0}(\gamma)$ are probabilities of symbols before the transformation is applied and let $p_{0,max}=max\{p_{0}(\beta),p_{0}(\gamma)\}$. After the transformation, the associated probabilities are $p_{1}(\beta)$, $p_{1}(\gamma)$ and $p_{1,max}=max\{p_{1}(\beta),p_{1}(\gamma)\}$. If $p_{1,max} > p_{0,max}$ then the generalized context transformation $T$ reduces entropy.
\end{theorem}

\begin{proof}
Let $c=p_{0}(\alpha)+p_{0}(\beta)=p_{1}(\alpha)+p_{1}(\beta)$, then since the entropy function of two different letters is defined on the interval $\langle 0;c\rangle$ and it is concave with maximum at $c/2$ and minimum at $0$ and $c$, then $p_{0,max}$ has to be located on the interval $p_{0,max}\in\langle c/2;c\rangle$, but on that particular interval the higher the maximum is the lower the entropy is, so if we increase the maximum(or we can say increase the absolute value of difference  $|p_{\beta}-p_{\alpha}|$), then the entropy will decrease.
\end{proof}

\clearpage
\Section{Algorithm}
\label{sec:algorithm}

Proposed algorithm aims to reduce entropy by sorting all conditional probabilities, such that the symbol's ordering, based on their probabilities, is approximately the same also in cases of conditional probabilities. For instance if the space character is the most frequent character in the text, then after the transformation, it should be, at least approximately, also the most frequent character if we consider the conditional distribution of symbols following arbitrary prefix $w$. The algorithm consists of three stages:

\begin{itemize}
\item Initial - collection of involved statistics.
\item Search for and application of generalized context transformations.
\item Storage.
\end{itemize}

\SubSection{Initial phase}

In the initial phase the file is passed once. In this pass the algorithm counts frequencies of all symbols. These frequencies are then used to sort symbols yielding an ordered set of indices $s$ to the alphabet. In proposed algorithm alphabet symbols are one byte values.

\begin{algorithm}
\caption{Preparation for search and application stage}
\label{alg:prep}
\begin{algorithmic}
\Function{HighOrderCT}{}
\State $contextTree \gets rootNode$
\For{$i=0;i<|\Sigma|;i{+}{+}$}
\State $positions \gets CollectPositions(s[i])$
\State $child \gets contextTree.AppendChild(s[i])$
\State $RecursiveCT(child,positions)$
\EndFor
\EndFunction
\end{algorithmic}
\end{algorithm}

Algorithm 1 loops through all symbols from the most frequent one to the least frequent one. Each time it passes through input message(file) $m$, it collects positions $t_i=\{j|m[j] = s[i]\}$ of symbol $\alpha = s[i]$ and builds a frequency table $s_\alpha$ of symbols at positions $t_{i,j}+1$ following the symbol $\alpha$. The complexity of this part is $O(|\Sigma|n)$ but in exchange, as will be discussed in the following section, we restrict every other search for and application of transformation onto the space of positions $t_i$ and in a function $RecursiveCT$ we won't need any other pass through the whole message, but only through the set of positions $t_i$.

\SubSection{Context transformation by maximal entropy reduction}

Suppose that the $RecursiveCT$ function has been called with a symbol $s[i]=\alpha$, then  we can form an ordered set $s_\alpha[j]$ of symbols ordered by conditional frequencies(or probabilities) $f(X|\alpha)$. 
In a final state, we would like to have a message in a state, such that for all prefixes $\alpha w$, where for all substrings $w$: $w_i\neq \alpha$, the same symbols are ordered at the same positions: $s_{\alpha w}[j] = s[j]$.
The $RecursiveCT$ function is given in the listing of Algorithm 2.

\begin{algorithm}
\caption{Recursive search for and application of context transformation}
\label{alg:max-ent-rec}
\begin{algorithmic}
\Function{RecursiveCT}{$parentNode, positions$}
\State $\alpha \gets parentNode.Symbol()$
\Repeat
\State $\beta,\gamma,dH \gets SearchCT(\alpha)$
\If{$dH > lim$}
\State \textbf{break}
\EndIf
\State $node \gets parentNode.AppendChild(\beta, \gamma)$
\State $pos \gets ApplyCT(\alpha,\beta,\gamma)$
%\State $pos\beta,pos\gamma \gets ApplyCT(context\_symbol,\beta,\gamma)$
\State $RecursiveCT(node, pos)$
%\State $RecursiveCT(node\gamma, pos\gamma)$
\Until{$true$}
\EndFunction
\end{algorithmic}
\end{algorithm}

The recursive algorithm searches for transformations until it is able to found a transformation that reduces entropy more, than the limit given by the $lim$ variable. The search is implemented in the function $SearchCT$ given in the listing of Algorithm 3. 
When the transformation is found, it is applied by the function $ApplyCT$, this function exchanges all symbols $\beta$ and $\gamma$ found by $SearchCT$ and if the number of occurences of symbol $\gamma$ is larger than of the symbol $\beta$, then it returns former positions of symbol $\gamma$, otherwise it returns former positions of symbol $\beta$. These positions will be later used in a transformation of higher order contexts.

The structure of applied transformations is stored as a tree, the so called context transformation tree. The root node of the tree doesn't represent any transformation and its children, called context symbol nodes, are individual context symbols $s[i]$ selected in Algorithm 1. Each transformation is stored as a new child node of the input parameter $parentNode$.

In the last step the function $RecursiveCT$ is called again. A context symbol of this call is the more frequent character and its corresponding set of positions. The set of positions $pos\beta$ resp. $pos\gamma$ are former positions of symbols $\gamma$ resp. $\beta$ in context of symbol $\alpha$. When the function $RecursiveCT$ is called for the first time and the transformation has been found, then its application corresponds to the mutual exchange of digrams $\alpha\beta$ and $\alpha\gamma$ and vice versa in the former message. Each other call of $RecursiveCT$, made from itself, gets into the higher order context. For instance the first call of $RecursiveCT(node, pos)$: suppose that the next transformation symbols, that has been found, are $\delta$ and $\epsilon$, then they correspond to the exchange of trigrams of the form $\alpha\gamma\delta$ and $\alpha\gamma\epsilon$ in the former text. Due to the fact that we are collecting positions of each replacement, the longer the context of transformation is, the smaller is the space(the size of the set of positions) where the transformation is applied.

\begin{algorithm}
\caption{Recursive search and application of context transformation}
\label{alg:max-ent-rec}
\begin{algorithmic}
\Function{SearchTransformation}{$context\_symbol$}
\State $dH \gets 0$
\State $\beta \gets context\_symbol$
\State $\gamma \gets context\_symbol$
\For{$i=0$ to $|\Sigma|$}
\If{$s[i] == context\_symbol}$
\State \textbf{continue}
\EndIf
\For{$j=i+1$ to $|\Sigma|$}
\State $temp\_dH = EntropyChange(s[i],s[j])$
\If{$temp\_dH < dH\ AND\ s[j] \neq context\_symbol$}
\State $dH = temp\_dH$
\State $\beta = s[i]$
\State $\gamma = s[j]$
\EndIf
\EndFor
\EndFor
\State \Return $[\beta,\gamma, dH]$
\EndFunction
\end{algorithmic}
\end{algorithm}

The $SearchTransformation$ function returns the maximally entropy reducing transformation. The $EntropyChange$ function in the listing of Algorithm 3 computes $dH$ using equation (4).
The search omits transformations when $s[j] = \alpha$ and frequency $f(\alpha,w,\alpha)\neq 0$, this is very important, because this condition ensures the existence of inverse transformation. The context symbol $\alpha$ is a symbol found in the first phase of the search algorithm and it is a symbol from which the transformation begins. Before the transformations starts, all positions of context symbol are found, if we would allow transformation of context symbol, then it would become impossible to distinguish between $\alpha$ that emerges due to the transformation and $\alpha$ in the former message. On the other hand, it is possible to introduce $\alpha$ into transformation if $f(\alpha,w,\alpha) = 0$, i.e. no word of the form $\alpha w \alpha$ is a substring of the former message.

There are situations when the change of symbols ordering occurs and even it reduces entropy. Suppose two consecutive symbols of ordered set $s$: $s[i]$ and $s[i+1]$, given that probabilities  $p(s[i])>p(s[i+1])$, let $\beta=s[i]$ and $\gamma=s[i+1]$ and context symbol to be $\alpha$, if $p(\alpha,\beta) - p(\alpha, \gamma) > p(\beta) - p(\gamma)$ then it is convenient to apply transformation $GCT(\alpha\beta\leftrightarrow\alpha\gamma)$ even though $p(\alpha,\beta)>p(\alpha,\gamma)$ and $p(\beta)>p(\gamma)$. Such transformation would eventually switch the order of symbols $\beta$ and $\gamma$ and reduces entropy. As a conclusion we remark that symbols sorted into the order given in the initial phase do not neccessarily have to have the same order in the final state.

\SubSection{Storage}
The result of transformation is stored in two parts, the first part contains description of context transformation tree and the second part contains transformed message. The transformed message has the same size as the input message, the additional overhead is given by the first part. The example of the resulted context transformation tree is visualized in Figure 1. The forward transformation is represented by paths from the root to the leaf node beginning with the most left path and finishing with the most right path.

\begin{figure}[ht!]
\centering
\includegraphics[width=4.0in]{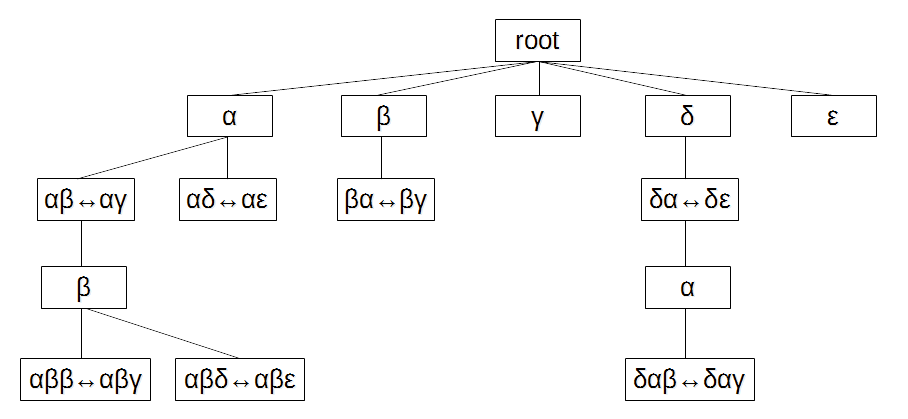}
\caption{\label{fig:3t}%
Context transformation tree corresponding to consecutive transformations: $\alpha\beta\leftrightarrow\alpha\gamma$ followed by two higher order transformations of trigrams $\alpha\beta\beta\leftrightarrow\alpha\beta\gamma$ and $\alpha\beta\delta\leftrightarrow\alpha\beta\epsilon$, the next transformation found and applied is $\alpha\delta\leftrightarrow\alpha\epsilon$. Then the second most frequent symbol is selected and the transformation $\beta\alpha\leftrightarrow\beta\gamma$ is applied. For the third and fifth context symbols $\gamma$ and $\epsilon$ there are no entropy reducing transformations. For the transformation $\delta\alpha\leftrightarrow\delta\epsilon$ of the fourth context symbol $\delta$ there is one higher order transformations $\delta\alpha\beta\leftrightarrow\delta\alpha\gamma$.}
\label{fig_sim}
\end{figure}

The tree is stored using recursive function $SaveHeader(rootNode)$ given in listing of Algorithm 4. In our implementation arguments of $WRITE$ function are one byte values.

\begin{algorithm}[!ht]
\begin{algorithmic}
\Function{SaveHeader}{$treeNode$}
\State $WRITE(treeNode.ChildrenLength())$
\For{$i=1$ to $treeNode.ChildrenLength()$}
\If{$treeNode.IsContextSymbolNode == \textbf{true}$}
\State $WRITE(node.Child(i).\alpha)$
\Else
\State$ WRITE(treeNode.Child(i).\beta)$
\State $WRITE(treeNode.Child(i).\gamma)$
\EndIf
\State $SaveHeader(treeNode.Child(i))$
\EndFor
\EndFunction
\end{algorithmic}
\caption{\label{fig:gct}%
Recursive header saving algorithm}
\end{algorithm}

\SubSection{Inverse transformation}

The inverse transformation $GCT^{-1} = GCT_{\rightarrow}$ is applied from the beginning to the end of a transfomed message. Transformations contained in the context transformation tree are applied in reversed order, beginning with the most right path and finishing with the most left path from the root node to the leaf node. As we briefly mentioned in the section about the algorithm searching the entropy reducing transformations, the algorithm transforms words $w_j$ residing between two consecutive context symbols $\alpha$. Actualy the algorithm behaves exactly like if we split the message on parts separated by context symbol $\alpha$ and we would transform each word independently of each other. The inverse transformation can be viewed from the same perspective, we found the first occurence of context symbol $\alpha$ and at the position that follows $\alpha$ we apply transformation from the tree in reversed order. Since we know that the transformation has been taken over symbols between two consecutive symbols $\alpha$, we can apply inverse transformation and the next occurence of symbol $\alpha$ will be the next position where the inverse transformation will be applied again.

\Section{Transformation of languages to languages with lower entropy}
\label{sec:languages}

Usually the most frequent character in a text is a space character separating individual words in a sentence.
Figure \ref{fig:book1-text-cut}.b) gives the example of a message after transformation of words residing between two consecutive space characters. Several words beginning on the new line, i.e. words that follows the end of line character, remains untransformed and they will be transformed by some later transformation. Figure \ref{fig:book1-text-cut}.c) then presents final state of the message. 

\begin{figure}[ht!]
\centering
\includegraphics[width=5.5in]{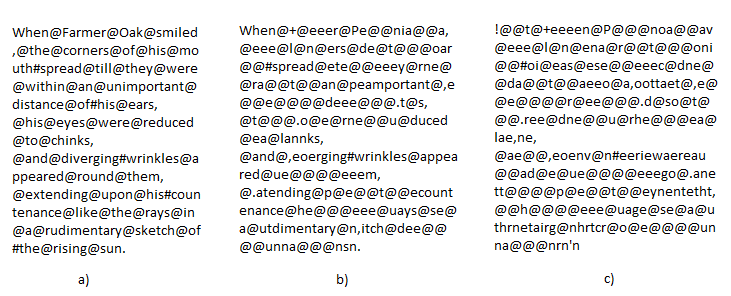}
\caption{\label{fig:book1-text-cut}%
Slice from the Calgary corpus book1 file, a) before transformation, b) after 10k transformations, c) final state(approx. 20k transformations). To emphasize differences between states of the message, all space characters has been replaced by @ and the end of line characters by \#.}
\end{figure}

The search algorithm can be modified to transform words residing between arbitrary distinct symbols. Suppose that we would like to transform all tags in HTML document, such that symbols between `\textless' and `\textgreater' become ones of low entropy. In the listing of Algorithm 3. the initial context symbol would be `\textless' and the transformation stops, when no other entropy reducing transformation exists than the one modifying occurence of  `\textgreater'. Rules of this kind can be very simply integrated into the search algorithm.

\Section{Results}
\label{sec:results}

Our results are summarized in Table 1. We haven't focused at the compression of the context transformation tree yet, instead we estimated upper bound of entropy of the tree's description using the PAQ8 algorithm created by Matt Mahoney \cite{paq8}. 
%The resulted value of bits per byte for each file  has been computed in the following way:
% as a sum of the size of the tree after compression by PAQ8 and the size of the transformed file given by entropy multiplied by the initial size of the file, the resulted sum is divided by the size of the resulted file.

Let $H_t$ is entropy of the transformed message $m$ of the size $|m|$, let $|h_{paq8}|$ is a size of tree compressed by PAQ8 in bits, then the bits per symbol(byte) ratio $F$ of the resulted message is computed as:

\[
F = \frac{|h_{paq8}| + H_t|m|}{|m|}
\]

The bits per byte ratio has been computed for several settings of the $lim$ variable given in the listing of Algorithm 2. The value of the limit is given in the subscript of the column name and represents number of bits, for instance $H_{2}$, resp. $F_{2}$, means entropy, resp. bits per byte ratio, for $lim=2$.

\begin{table}[!ht]
\begin{center}
{
\renewcommand{\baselinestretch}{1}\footnotesize
\begin{tabular}{|c|c|c|c|c|c|c|c|c|c|}
\hline 
File  & $H$ & $H_{0}$ & $H_{4}$ & $H_{8}$ & $F_{0}$ & $F_{4}$ & $F_{8}$ & gzip & bzip2\tabularnewline
\hline 
\hline 
bib & 5.201 & 2.355 & 2.938 & 3.214 &3.545 & 3.529 & 3.551 & 2.509 & \textbf{1.975}\tabularnewline
\hline 
book1 & 4.527 & 3.001 & 3.318 & 3.414 & 3.764 & 3.581 & 3.552 & 3.250 & \textbf{2.420}\tabularnewline
\hline 
obj1 & 5.948 & \textbf{1.347} & 2.419 & 3.844 & 5.572  & 5.650 & 4.945 & 3.812 & 4.015\tabularnewline
\hline 
paper1 & 4.983 & \textbf{2.316} & 3.058 & 3.344 & 4.062 & 3.919 & 3.840 & 2.789 & 2.492\tabularnewline
\hline
paper2 & 4.601 & 2.471 & 3.019 & 3.273 & 3.765 & 3.587 & 3.581 & 2.887 & \textbf{2.437}\tabularnewline
\hline 
progc & 5.199 & \textbf{2.346} & 3.071 & 3.421 & 4.336& 4.103 & 3.989 & 2.677 & 2.533\tabularnewline
\hline 
progp & 4.868 & 1.766 & 2.296 & 2.683 & 3.242 & 3.299 & 3.266 & 1.811 & \textbf{1.735}\tabularnewline
\hline 
trans & 5.532 & \textbf{1.473} & 2.289 & 2.667 & 2.784 & 3.141 & 3.205 & 1.610 & 1.528\tabularnewline
\hline 
\hline
alice29.txt & 4.567 & 2.608 & 2.971 & 3.141 & 3.578 & 3.387 & 3.372 & 2.850 & \textbf{2.272}\tabularnewline
\hline
bible.txt  & 4.342 & 2.662 & 2.727 & 2.757 & 2.756 & 2.762 & 2.762 & 2.201 & \textbf{1.672} \tabularnewline
\hline
cp.html & 5.229 & \textbf{1.593} & 2.767 & 3.221 & 4.248 & 3.992 & 3.817 & 2.593 & 2.479\tabularnewline
\hline
kennedy.xls & 3.573 & 3.143  & 3.146 & 3.150 & 3.164 & 3.158 & 3.156 & 1.629 & \textbf{1.012}\tabularnewline
\hline
world192 & 4.998 & 2.617 & 2.803 & 2.931 & 3.094 & 3.037 & 3.057 & 2.259 & \textbf{1.583}\tabularnewline
\hline
\end{tabular}
}
\caption{\label{tab:res2}%
The summary of achieved results for selected files from Calgary\cite{calgary} and Canterbury\cite{canterbury} corpuses in comparison with standard compression methods of gzip and bzip2. $H$ represents entropy of input file. Values of $F_{lim}$ and results of standard methods are measured in bits per byte.}
\end{center}
\end{table}

The case $H_0$ is the extreme case when all accessible entropy reducing transformations were applied and in several cases the achieved entropy rate was better than the one achieved by standard methods, but the large number of transformations leads to a growth of context transformation tree and as a consequence the resulted file size is significantly larger. The difference $F_{lim} - H_{lim}$ gets smaller with larger files(bible.txt - 3.85 MB) because of the relative size of the tree against the size of the message, but also due to the observation that larger files are less sensible to the selection of the value of entropy reduction limit $lim$.

\Section{Conclusion}
\label{sec:conclusion}
We proposed a transformation of higher order contexts based on the concept of generalized context transformations. Our algorithm is
able to significantly reduce entropy of input messages, but it is only of limited ability to efficiently store the information that is neccessary to store description of all transformations. The efficient storage of context transformation tree will be one of the areas we will focus at in the future, and it is fair to say, that it is a main weakness of our algorithm. On the other hand, this issue partially relates to the size of the input message, larger files like bible.txt from Canterbury Corpus have relatively small context transformation tree in comparison with the overall size of input. Under assumption, that the total number of applied transformations grows logarithmically with the size of the file, then we assume that the effect of the tree storage, on the resulted bits per byte ratio, should be decreasing with increasing size of file.

\Section{References}
\bibliographystyle{IEEEtran}
\bibliography{refs}

\end{document}